\documentclass[11pt]{article}
\usepackage[margin=1in]{geometry}

\usepackage[pdftex]{graphicx}
\graphicspath{{../pdf/}{../jpeg/}}
\DeclareGraphicsExtensions{.pdf,.jpeg,.png}
\usepackage{amssymb,amsmath,url}
\usepackage{bm} 
\usepackage{multirow}
\usepackage{booktabs}	
\usepackage{epstopdf}
\usepackage{color}
\usepackage{subfigure}
\usepackage{bigfoot}		
\usepackage{algpseudocode}
\usepackage{pifont}
\usepackage{tikz}
\usepackage{color}
\usepackage{cite}
\usepackage{lipsum}

\newtheorem{theorem}{Theorem}

\newtheorem{proposition}{Proposition}

\allowdisplaybreaks
\newtheorem{definition}{Definition} 
\newtheorem{proof}{Proof}

\newcommand{\R}{\mathbb{R}}

\newcommand{\X}{\mathbb{X}}

\newcommand*{\QEDA}{\hfill\ensuremath{\blacksquare}}

\DeclareMathOperator*{\argmax}{arg\,max}
\title{\LARGE
	Exploiting Weak Supermodularity\\ for Coalition-Proof Mechanisms\footnotetext{This research was gratefully funded by the European Union ERC Starting Grant CONENE.}\\
}
\author{Orcun Karaca\thanks{The authors are with the Automatic Control Laboratory, Department of Information Technology and Electrical Engineering, ETH Z\"{u}rich, Switzerland. e-mails: {\tt\small \{okaraca, mkamgar\}@control.ee.ethz.ch}} \and Maryam Kamgarpour\footnotemark[1]
}
\begin{document}
\maketitle
\begin{abstract}\noindent
 Under the incentive-compatible Vickrey-Clarke-Groves mechanism, coalitions of participants can influence the auction outcome to obtain higher collective profit.  These manipulations were proven to be eliminated if and only if the market objective is supermodular. Nevertheless, several auctions do not satisfy the stringent conditions for supermodularity. These auctions include electricity markets, which are the main motivation of our study. To characterize nonsupermodular functions, we introduce the supermodularity ratio and the weak supermodularity. We show that these concepts provide us with tight bounds on the profitability of collusion and shill bidding. We then derive an analytical lower bound on the supermodularity ratio. Our results are verified with case studies based on the IEEE test systems.

\end{abstract}

\section{Introduction}

Over the last couple of decades, electricity markets have been undergoing a rapid transformation from tightly regulated monopolies to deregulated competitive market structures~\cite{wilson2002architecture}. This restructuring has been essential to improve economic efficiency and attract new investments~\cite{cramton2017electricity}. Designing electricity markets is not a simple task since supply and demand of electricity need to be balanced in every instance of time. Specifically, high penetration of intermittent renewable energy sources presents challenges in maintaining this stability~\cite{bosesome}. Hence, there has been a surge of interest from the control community in studying various electricity markets \cite{roozbehani2010stability, tang2016model,orcun2018game,karaca2017game}.

This work studies a subset of the existing market mechanisms, conducted as reverse auctions. In these markets, generators submit their bids, and then an independent system operator determines the power allocation and the payment for each generator. The allocation rule is the economic dispatch which secures a reliable operation. Then, the central element of the market design is the payment rule, since the generators have incentives to strategize around~it. In~particular, the operator designs the payment rule to ensure that the generators reveal their true costs in order to achieve a stable grid with maximum social welfare~\cite{cramton2017electricity}.

Previous work on electricity markets considered the pay-as-bid~\cite{orcun2018game} and the locational marginal pricing mechanisms~\cite{wu1996folk}. In both mechanisms, generators can bid strategically to influence their profits since these mechanisms do not incentivize truthful bidding. Studies have shown that strategic manipulations have increased electricity prices substantially in these markets~\cite{wolfram1997strategic,joskow2001quantitative}. As an alternative, under the Vickrey-Clarke-Groves (VCG) mechanism, truthful bidding is the dominant-strategy Nash equilibrium~\cite{vickrey1961counterspeculation,clarke1971multipart,groves1973incentives}. As a result, every generator finds it more profitable to reveal their true costs, regardless of the bids of other generators. Due to this property, several recent works proposed the use of the VCG mechanism in electricity markets \cite{pgs,xu2017efficient}. 

Despite the desirable theoretical properties of the VCG mechanism, coalitions of generators can strategically bid to increase their collective utility. Hence, this mechanism is susceptible to collusion and shill bidding \cite{ausubel2006lovely}. Because the same market participants are involved in similar transactions day after day, electricity markets can be exposed to such manipulations~\cite{anderson2011implicit}. This is crucial since in a larger context the VCG~mechanism is not truthful.

As is outlined in auction literature \cite{ausubel2006lovely}, collusion and shill bidding occur when the VCG outcome is not in the~\textit{core}. The core is a concept from coalitional game theory where the participants have no incentive to leave the coalition of all participants. As is proven in~\cite{karaca2017game}, the VCG outcomes lie in the core if and only if the market objective is supermodular. Supermodularity can only be achieved in restricted settings, such as polymatroid constraints and convex bids. Consequently, electricity markets are generally not supermodular since technical rigidities result in complex constraint sets.

Our goal is to understand the properties of the VCG mechanism if the stringent conditions for supermodularity do not hold. To identify the bottlenecks caused by strategic manipulations, we aim to provide bounds on the profitability of collusion and shill bidding by defining the concept of weak supermodularity. For this concept, we are inspired by the growing literature on the submodularity ratio~\cite{bian2017guarantees,das2011submodular,elenberg2016restricted}.

Our contributions in this paper are as follows. First, we introduce a new way to characterize nonsupermodular set functions, that is, the supermodularity ratio. This ratio quantifies how close a function is to being supermodular. Second, we show that the supermodularity ratio provides us with tight bounds on the profitability of collusion and shill bidding under the VCG mechanism. Third, we derive an analytical lower bound on the supermodularity ratio of the electricity markets under consideration. Finally, we verify our results with case studies based on the IEEE test systems. The results derived on collusion and shill bidding apply to general auctions run by the VCG mechanism. To the best of our knowledge, this is the first work providing such bounds for VCG outcomes not lying in the core.

The remainder of this paper is organized as follows. 
In Section~\ref{sec:mech}, we introduce the electricity market and discuss the VCG mechanism. Section~\ref{sec:ws} defines the supermodularity ratio. First, we obtain bounds on collusion and shill bidding, see Theorem~\ref{thm:approx_collusions}. We then provide a lower bound on the supermodularity ratio of the markets, see Theorem~\ref{thm:lowbnd}. In Section~\ref{sec:num}, we present the numerical results.
\section{Mechanism framework}\label{sec:mech}

We start with a generic electricity market reverse auction. The set of participants consists of the central operator $l=0$ and the bidders $L=\{1,\ldots,\lvert L\rvert\}$. Let there be $t$ types of power supplies. These types can include control reserves (positive, negative, secondary, tertiary) and power injections differentiated by their locations, durations, and scheduled times. The same type of supplies from different bidders are perfect substitutes for the central operator. We assume that each bidder~$l$ has a private true cost~function $c_l: \X_l \rightarrow \mathbb R_+$, where $0\in \X_l\subseteq\R_+^t$ and~$c_l(0)=0$. This last assumption, $c_l(0)=0$, holds for many electricity markets, for instance, control reserve markets, and day-ahead markets that include generators' start-up costs\cite{abbaspourtorbati2016swiss,xu2017efficient}. Each bidder~$l$ then submits a bid function to the central operator, denoted by $b_l:\hat \X_l \rightarrow \mathbb R_+$, where $0\in \hat \X_l\subseteq\R_+^t$ and~$b_l(0)=0$. As is discussed in \cite[\S3.1]{karaca2017game}, these bid functions also capture the traditional multiple-item auction setup.

Given the bid profile $\mathcal{B}=\{b_l\}_{l\in L}$, \textit{a mechanism} defines an allocation rule $x_l^*(\mathcal{B})\in\hat \X_l$ and a payment rule $p_l(\mathcal{B})\in\R$ for each bidder $l$. 
In many electricity markets, the allocation is determined by the economic dispatch, that is, minimizing the aggregate cost subject to some constraints
\begin{equation}\label{eq:main_model}
\begin{split}
J(\mathcal{B})=&\min_{x\in\hat \X,y}\, \sum\limits_{l\in L} b_l(x_l) + d(x,y)\\
&\ \ \mathrm{s.t.}\ \ g(x,y)\leq 0.\\
\end{split}
\end{equation}
Here~$y\in\R^p$ are variables entering~\eqref{eq:main_model} in addition to the allocation~$x\in\hat \X=\prod_{l\in L} \hat \X_l$. The function~$d:\R^{t\rvert L\rvert}_+\times\R^p\rightarrow \R$ is an additional cost term. For example, in the case of a two-stage electricity market, $y$ corresponds to the second stage variables and $d$ is the second stage cost. The function $g:\R^{t\rvert L\rvert}_+\times\R^{p}\rightarrow \R^{q}$ defines the constraints.\footnote{Problem (\ref{eq:main_model}) defines a general class of markets, for example, energy~and~reserve markets\cite{xu2017efficient, carlson2012miso} and stochastic markets~\cite{abbaspourtorbati2016swiss,conejo2010decision}.} Finally, if the problem is infeasible, the objective is $J(\mathcal{B})=\infty$.

Let the optimal solution be denoted by $x^*(\mathcal{B})\in\hat \X$ and $y^*(\mathcal{B})\in\R^{p}_+$.\footnote{We assume that in case of multiple optima, there is some tie-breaking rule according to a predetermined fixed ordering of the bidders.} The \textit{utility} of bidder $l$ is $$u_l(\mathcal{B})=p_l(\mathcal{B})-c_l(x^*_l(\mathcal{B})).$$ A bidder whose bid is not accepted is not paid and $u_l(\mathcal{B})=0$. Then, the total payment of the central operator is given by  $$u_0(\mathcal{B})=-\sum_{l\in L} p_l(\mathcal{B}) - d(x^*(\mathcal{B}),y^*(\mathcal{B})),$$
which defines the utility of the central operator.  Note that this payment can be an expected value when the function $d$ is an expected second stage cost.

Three fundamental properties we desire in mechanism design are individual rationality, efficiency and dominant-strategy incentive-compatibility. A mechanism is \textit{individually rational} if bidders do not face negative utilities, $u_l(\mathcal{B})\geq 0$, $\forall l\in L$.\footnote{It is also often referred to as voluntary participation or cost recovery.} A mechanism is \textit{efficient} if the sum of all utilities $\sum_{l=0}^{\lvert L\rvert} u_l(\mathcal{B})$ is maximized. This property is attained if we are solving for the optimal allocation of the market in~\eqref{eq:main_model} under the condition that the bidders submitted their true costs.

To define dominant-strategy incentive-compatibility, we first bring in tools from game theory. Let $\mathcal{B}_{-l}$ denote the bid profile of all the bidders, except bidder $l$.
The bid profile $\mathcal{B}$ is a \textit{Nash equilibrium} if for every bidder $l$, $u_l(\mathcal{B}_l,\mathcal{B}_{-l})\geq u_l(\tilde{\mathcal{B}}_l,\mathcal{B}_{-l})$, $\forall\tilde{\mathcal{B}}_l$.
The bid profile $\mathcal{B}$ is a \textit{dominant-strategy Nash equilibrium} if for every bidder $l$,
$u_l(\mathcal{B}_l,\hat{\mathcal{B}}_{-l})\geq u_l(\tilde{\mathcal{B}}_l,\hat{\mathcal{B}}_{-l})$, $\forall\tilde{\mathcal{B}}_l$, $\forall \hat{\mathcal{B}}_{-l}$. Finally, a mechanism is  \textit{dominant-strategy incentive-compatible} (DSIC) if the truthful bid profile $\mathcal C=\{c_l\}_{l\in L}$ is the dominant-strategy Nash equilibrium. In other words, every bidder finds it more profitable to bid truthfully,~regardless of what others bid.

\subsection{Payment rules}

The design of the payment rule plays a crucial role in attaining the desirable properties above.
We first consider two prominent payment rules widely used for the electricity markets, that is, the pay-as-bid mechanism and the locational marginal pricing (LMP) mechanism. 

Under the pay-as-bid mechanism, a rational bidder would overbid to ensure positive utility. There exist many Nash equilibria, none of which are incentive-compatible \cite{orcun2018game}. Under the LMP mechanism,  strategic manipulations become more complex than the case of the pay-as-bid mechanism. A~bidder can maximize its utility by both inflating the bids and withholding its maximum supply~\cite{ausubel2014demand}. Furthermore, an equilibrium may not even exist~\cite{tang2013nash}. In summary, none of these payment rules satisfy the properties of efficiency and incentive-compatibility. 

The \textit{Vickrey-Clarke-Groves (VCG) mechanism} is characterized by the payment rule, $
p_l(\mathcal{B})=b_l(x^*_l(\mathcal{B}))+(h(\mathcal{B}_{-l})-J(\mathcal{B}))$.
The function $h(\mathcal{B}_{-l})\in\R$ must be carefully chosen to ensure individual rationality. We use the \textit{Clarke pivot rule} $
h(\mathcal{B}_{-l})=J(\mathcal{B}_{-l}),  $
where $J(\mathcal{B}_{-l})$ denotes the minimum total cost without bidder $l$, that is, the optimal value of the optimization problem in (\ref{eq:main_model}) with $x_l=0$.\footnote{This rule generates the minimum total payment by the central operator  while ensuring individual rationality of the bidders\cite[Theorem~1]{krishna1998}.} Note that this mechanism is well-defined if a feasible solution exists when a bidder is removed. This is not restrictive in the presence of many bidders and a second-stage market. 

Our previous work in \cite{karaca2017game} shows that the VCG mechanism satisfies all three fundamental properties for the model introduced in~\eqref{eq:main_model}. This result is a generalization of \cite{vickrey1961counterspeculation,clarke1971multipart,groves1973incentives} which do not consider continuous goods, second stage cost and general constraints. Dominant-strategy incentive-compatibility makes it easier for bidders to enter the auction, without spending resources in computing optimal bidding strategies.

Despite these remarkable theoretical virtues, the VCG mechanism can suffer from collusion and shill bidding~\cite{ausubel2006lovely}. Bidders $K\subseteq L$ are \textit{colluders} if they obtain higher collective utility by changing their bids from $\mathcal C_K=\{c_l\}_{l\in K}$ to $\mathcal{B}_K=\{b_l\}_{l\in K}$. A bidder $l$ is a \textit{shill bidder} if there exists a set~$S$ and bids $\mathcal{B}_S=\{b_k\}_{k\in S}$ such that the bidder~$l$ finds participating with multiple bids $\mathcal{B}_S$ more profitable than participating with a single truthful bid $\mathcal{C}_l$. These shortcomings are illustrated by electricity market examples in \cite{pgs,karaca2017game}.

This observation motivates us to define \textit{coalition-proofness}. By coalition-proof, we mean that a group of bidders, who lose when bidding their true cost, cannot profit from collusion, and no bidder can profit from bidding with multiple identities. We remark that it is not possible to achieve immunity to collusion from all sets of bidders, see also the examples in~\cite{beck2009revenue}. For instance, no mechanism can eliminate the case where all bidders inflate their bid prices simultaneously. 

\subsection{Core as a coalition-proof outcome}
In coalitional game theory, the \textit{core} defines the set of utilities that cannot be improved upon by forming coalitions.\footnote{The utility allocation and the auction outcome are used interchangeably.} We discuss how this concept coincides with coalition-proofness.

For every $S\subseteq L$, let $J(\mathcal{B}_S)$ be the optimal value of~(\ref{eq:main_model}) with $x_{-S}=0$,
where the stacked vector $x_{-S}\in\R_+^{t(\lvert L\rvert -\lvert S\rvert)}$ is defined by omitting the subvectors from~$S$. In particular, it is defined by the following expression:
\begin{equation*}
\begin{split}
J(\mathcal{B}_S) =  &\min_{x,y}\sum\limits_{l\in S} b_l(x_l) + d(x,y)\\
&\ \mathrm{  s.t. } \ \,g(x,y)\leq 0,\, x_{-S}=0.
\end{split}
\end{equation*}Note that, this function is nonincreasing, $J(\mathcal{B}_R) \leq J(\mathcal{B}_S)$, for~$S\subseteq R$. We refer to $J$ as the market objective.
Then, the core $Core(\mathcal{C})\in\R\times\R^{\rvert L\rvert}_+$ is defined as
	\begin{equation}\label{eq:mcoredef}
	\begin{split}
	\Big\{u\in\R\times\R^{\rvert L\rvert}_+ \,\rvert\, &u_0+\sum\limits_{l\in L} u_l=-J(\mathcal{C}),\\
	&u_0+\sum\limits_{l\in S}u_l\geq-J(\mathcal{C}_S),\, \forall S \subset L \Big\}.
		\end{split}
	\end{equation}
The core is always nonempty in auctions because the outcome $u_0=-J(\mathcal{C})$, $u_l=0$ for all $l\in L$ always lies in the core. Core outcomes are individually rational since they are restricted to the nonnegative orthant for the bidders.  The equality constraint in \eqref{eq:mcoredef} implies that the outcomes are efficient since the term on the right is maximized.
We say that an outcome is unblocked if no set of bidders can make a deal with the operator from which everyone can benefit. This condition is ensured by the inequality constraints in the definition of the core.

As is outlined in auction literature \cite{ausubel2006lovely} and extended to the market setup \eqref{eq:main_model} in \cite[Theorem~3]{karaca2017game}, the VCG mechanism is coalition-proof if and only if the VCG outcomes lie in the core. Hence, past work considers characterizing bid curves and constraints such that the VCG outcome is guaranteed to lie in the core. 
To state these results and extend beyond them, we bring in the definition of supermodularity. For the following definitions, we use $J(S)$ instead of $J(\mathcal{B}_S)$ since the properties are required to hold under any bid profile $\mathcal{B}$. We denote the set-theoretic difference $S\setminus \{l\}$ by $S_{-l}$.

\begin{definition}\label{def:sup}
	A function $J:2^{L}\rightarrow\R$ is \textit{supermodular} if  
	$J(S)-J(S_{-l})\leq J(R)-J(R_{-l})$
	for all coalitions $S\subseteq R\subseteq L$ and for each bidder $l\in S$. A function $\hat J:2^{L}\rightarrow\R$ is \textit{submodular} if and only if $-\hat J$ is supermodular.
\end{definition}

For any set of bids, the VCG outcome is in the core if and only if the market objective~$J$ is supermodular~\cite[Theorem~2]{karaca2017game}. Supermodularity of problem~\eqref{eq:main_model} is a strong condition which can only be achieved in restricted settings, such as polymatroid-type constraints and convex bids~\cite[Theorem~5]{karaca2017game}. Though the convex bid assumption may be reasonable in certain markets, the polymatroid constraint requirement is stringent. In particular, DC-OPF problems involve polytopic constraints and a polytope is in general not a polymatroid. In such instances, the market objective is generally not supermodular. As a result, the VCG mechanism suffers from collusion and shill bidding. To the best of our knowledge,  there does not exist any result bounding the profitability of such manipulations, in case the objective function is not supermodular. Therefore, we aim to quantify the coalition-proofness property of the VCG mechanism under more general bid functions and constraints. As a remark, all the following results hold for any general auction over continuous (or discrete) goods with complex constraints.

\section{Approximating Coalition-Proofness}\label{sec:ws}

For the results in this section, we are inspired by the growing literature on the submodularity ratio~\cite{bian2017guarantees,das2011submodular}. We first define the supermodularity ratio to quantify how close a set function is to being supermodular. 
Using the supermodularity ratio, we then introduce the weak supermodularity condition for the market objective in problem~\eqref{eq:main_model}. Finally, we show that this condition indeed provides us with an approximate coalition-proofness certificate for the VCG mechanism.

\begin{definition}\label{def:suprat}
	The \textit{supermodularity ratio} of a nonnegative set function $J:2^{L}\rightarrow\R_+$ is the largest $\gamma_{\text{sup}}$ such that
	\begin{equation*}
	\gamma_{\text{sup}}\Big[\sum_{l \in K }  J(S_{-l}) - J(S)\Big]\leq J(S_{-K}) - J(S),\, \forall K,S\subseteq L.
	\end{equation*}
\end{definition}

In literature, submodularity ratio is defined to quantify how close a nonnegative set function is to being submodular~\cite{das2011submodular,bian2017guarantees}. Since the function $J$ is supermodular if and only if $-J$ is submodular, it may seem natural to use the submodularity ratio to describe how close $-J$ is to being submodular. However, submodularity ratio is only defined over nonnegative set functions. Even if we allow for nonpositive functions, it does not provide any useful information for our purposes. In addition, this ratio can be zero in the cases where the supermodularity ratio is positive. As an alternative to the submodularity ratio, previous studies also discuss curvature to quantify how close a nondecreasing set function is to being supermodular. However, the objective in \eqref{eq:main_model} is nonincreasing and curvature is unbounded for nonincreasing set functions~\cite{bian2017guarantees}.

Motivated by the discussion above, we derive some important observations for the supermodularity ratio of nonincreasing set functions.

\vspace{.1cm}
\begin{proposition}\label{rem}
	The following statements hold for a nonincreasing set function: \begin{itemize}
		\item[(i)]  $\gamma_{\text{sup}}\in[0,\,1]$, 
		\item[(ii)] this function is supermodular if and only if $\gamma_{\text{sup}}=1$.  
	\end{itemize}
\end{proposition}
\vspace{.1cm}

The proof of the proposition is relegated to the appendix.
We define a set function to be  \textit{weakly supermodular} if $\gamma_{\text{sup}}>0$.

Next, we show that the weak supermodularity provides us with bounds on both shill bidding and collusion, hence, we achieve approximate coalition-proofness.
\vspace{.2cm}

\begin{theorem}
	\label{thm:approx_collusions}
	For the bidders $L$, consider a~VCG~mechanism modeled by \eqref{eq:main_model}. If the market objective~$J$ is weakly supermodular,~then, 
	\begin{itemize}
		\item[(i)] {A subset of bidders $K\subseteq L$ who lose when bidding their true values cannot profit more than \begin{equation*}
				[\gamma_{\text{sup}}^{-1}-1][J(\mathcal{C})-J(\mathcal{C}_{-K},\mathcal{B}_{K})],
			\end{equation*} by a joint deviation $\mathcal{B}_{K}$. For any joint deviation $\mathcal{B}_{K}$, this total profit is upper bounded by \begin{equation*} [\gamma_{\text{sup}}^{-1}-1][J(\mathcal{C})-J(\mathcal{C}_{-K},\mathcal{B}_{K}^0)],\end{equation*} where the bid profile $\mathcal{B}_{K}^0$ is the set of bids $b_l(x_l)=0,\,\forall x_l\in\X_l,\, \forall l\in K$.}
		\item[(ii)] For any bidder~$l$, the profit from bidding with the set of bids $\mathcal{B}_S$ is at most $$[\gamma_{\text{sup}}^{-1}-1][J(\mathcal{C}_{-l})-J(\mathcal{C}_{-l},\mathcal{B}_{S})],$$ 
		more than the profit from a single truthful bid $\mathcal{C}_l$. For any set of bids $\mathcal{B}_S$, this value is upper bounded by
		$$[\gamma_{\text{sup}}^{-1}-1][J(\mathcal{C}_{-l})-J(\mathcal{C}_{-l},\mathcal{B}_{l}^{0})],$$ where the bid $\mathcal{B}_{l}^0$ is given by $b_l(x_l)=0,\,\forall x_l\in\X_l$.\end{itemize}
\end{theorem}
\vspace{.2cm}

\begin{proof}
		(i) Let $K$ be a set of colluders who would lose the auction when bidding their true values $\mathcal{C}_K=\{c_l\}_{l\in K}$, while bidding $\mathcal{B}_K=\{b_l\}_{l\in K}$ they become winners, that is, they are all allocated a positive quantity. We then define $\mathcal{C}=(\mathcal{C}_{-K},\mathcal{C}_{K})$ and ${\mathcal{B}}=(\mathcal{C}_{-K},{\mathcal{B}}_{ K})$ where $\mathcal{C}_{-K}=\{c_l\}_{l\in L\setminus K}$ denotes the bidding profile of the remaining bidders. The profile $\mathcal{C}_{-K}$ is not necessarily a truthful profile. The total VCG utility that colluders receive under~${\mathcal{B}}$, $\sum_{l\in K}{u}_l^{\text{VCG}}({\mathcal{B}})$, is
	\begin{align*} &\hspace{-0.05cm}= \sum_{l\in K}J({\mathcal{B}_{-l}})-J({\mathcal{B}})+b_l(x_l^*({\mathcal{B}}))-{c}_l(x_l^*({\mathcal{B}}))\\
&\hspace{-0.05cm}\leq \gamma_{\text{sup}}^{-1}[J({\mathcal{B}}_{-K} ) - J({\mathcal{B}})]+\sum_{l\in K} b_l(x_l^*({\mathcal{B}}))-{c}_l(x_l^*({\mathcal{B}}))\\ 
&\hspace{-0.05cm} = [\gamma_{\text{sup}}^{-1}-1][J(\mathcal{C})-J(\mathcal{B})]+J({\mathcal{C}})- \Big[\sum_{l\in L\setminus K} c_l(x_l^*({\mathcal{B}}))  +\sum_{l\in K} b_l(x_l^*({\mathcal{B}})) +d(x^*({\mathcal{B}}),y^*({\mathcal{B}}))\Big]\\&\hspace{10.5cm}+\Big[\sum_{l\in K} b_l(x_l^*({\mathcal{B}}))-{c}_l(x_l^*({\mathcal{B}}))\Big] \\
&\hspace{-0.05cm} = [\gamma_{\text{sup}}^{-1}-1][J(\mathcal{C})-J(\mathcal{B})]+J({\mathcal{C}}) - \Big[\sum_{l\in L} {c}_l(x_l^*({\mathcal{B}}))+d(x^*({\mathcal{B}}),y^*({\mathcal{B}}))\Big] \\ &\hspace{-0.05cm}\leq [\gamma_{\text{sup}}^{-1}-1][J(\mathcal{C})-J(\mathcal{B})]\\
&\hspace{-0.05cm}\leq [\gamma_{\text{sup}}^{-1}-1][J(\mathcal{C})-J(\mathcal{C}_{-K},\mathcal{B}_{K}^0)].
\end{align*}
The first equality follows from the VCG payments. {The first inequality follows from the weak supermodularity.} The second equality comes from the fact that $K$ was a group of losers, {so $J(\mathcal{B}_{-K})=J(\mathcal{C}_{-K})=J({\mathcal{C}})$}. We also add and subtract the term $[J(\mathcal{C})-J(\mathcal{B})]$.
After substituting these, we see that the term in brackets is the cost of ${\mathcal{C}}$ but evaluated at a feasible suboptimal allocation $(x^*(\mathcal{B}),y^*(\mathcal{B}))$. Then, the second inequality follows from the fact that the term in the brackets is lower bounded by $J(\mathcal{C})$. This yields the first statement. Finally, the third inequality is obtained from $J(\mathcal{C}_{-K},\mathcal{B}_{K}^0)\leq J(\mathcal{B})$. This holds since $x^*(\mathcal{B})$ is a feasible suboptimal solution to $J(\mathcal{C}_{-K},\mathcal{B}_{K}^0)$. This concludes the proof of part (i).
	
	(ii) Similar to part (i), define $\mathcal{C}=(\mathcal{C}_{-l},\mathcal{C}_{l})$. The profile $\mathcal{C}_{-l}$ is not necessarily a truthful profile. Shill bids of bidder~$l$ are given by $\mathcal{B}_{S}=\{b_k\}_{k\in S}$.~We define a merged bid $\tilde{\mathcal{B}}_l$ as $$\tilde{b}_l(x_l)=\min_{x_k\in\hat\X_k,\,\forall k}\, \sum_{k\in S}b_k(x_k)\ \mathrm{s.t. }\sum_{k\in S}x_k=x_l.$$ We then define ${\tilde{\mathcal{B}}}=(\mathcal{C}_{-l},{\tilde{\mathcal{B}}}_{l})$.
	The VCG utility obtained from shill bidding under ${\mathcal{B}}=(\mathcal{C}_{-l},{\mathcal{B}}_{S})$, $\sum_{k\in S}{u}_k^{\text{VCG}}({\mathcal{B}})$, is
	\begin{align*} &= \sum_{k\in S}[J({\mathcal{B}_{-k}})-J({\mathcal{B}})+b_k(x_k^*({\mathcal{B}}))]-{c}_l(\sum_{k\in S}x_k^*({\mathcal{B}}))\\
	&\leq \gamma_{\text{sup}}^{-1}[J({\mathcal{B}}_{-S} ) - J({\mathcal{B}})]+\sum_{k\in S} b_k(x_k^*({\mathcal{B}}))-{c}_l(\sum_{k\in S}x_k^*({\mathcal{B}}))\\ 
	&= \gamma_{\text{sup}}^{-1}[J({\mathcal{C}}_{-l} ) - J({\tilde{\mathcal{B}}})]+\tilde{b}_l(\sum_{k\in S}x_k^*({\mathcal{B}}))-{c}_l(\sum_{k\in S}x_k^*({\mathcal{B}}))\\ 
	&= [\gamma_{\text{sup}}^{-1}-1][J({\mathcal{C}}_{-l}) - J({\tilde{\mathcal{B}}})]+ {u}_l^{\text{VCG}}({\tilde{\mathcal{B}}}) \\ 
	&\leq [\gamma_{\text{sup}}^{-1}-1][J({\mathcal{C}}_{-l}) - J(\mathcal{C}_{-l},\mathcal{B}_{S})]+ {u}_l^{\text{VCG}}({\mathcal{C}})\\
	&\leq [\gamma_{\text{sup}}^{-1}-1][J({\mathcal{C}}_{-l} ) - J(\mathcal{C}_{-l},\mathcal{B}_{l}^{0})]+ {u}_l^{\text{VCG}}({\mathcal{C}}). 
	\end{align*}
	The first inequality follows from the weak supermodularity of~$J$. The second equality holds since we have $J({\tilde{\mathcal{B}}})=J({{\mathcal{B}}})$. This follows from the definition of the merged bid and the following implication. Since the same type of supplies are perfect substitutes for the central operator, the functions $g$ and $d$ in fact depend on $\sum_{l\in L} x_l$. 
	 The third equality follows from adding and subtracting the term $[J({\mathcal{C}}_{-l} ) - J({\tilde{\mathcal{B}}})]$. The second inequality is the DSIC property of the VCG mechanism. This yields the first statement. Finally, the third inequality is obtained from $J(\mathcal{C}_{-l},\mathcal{B}_{l}^{0})\leq J({\tilde{\mathcal{B}}})=J(\mathcal{C}_{-l},\mathcal{B}_{S})$. This holds since $x^*({\tilde{\mathcal{B}}})$ is a feasible suboptimal solution to the problem defined by $J(\mathcal{C}_{-l},\mathcal{B}_{l}^{0})$. This concludes the proof of part (ii).  \QEDA
\end{proof}

Theorem~\ref{thm:approx_collusions} provides the first bounds on the profitability of collusion and shill bidding for the cases when the outcome of the VCG mechanism is not in the core. In summary, as the supermodularity ratio gets larger, the function $J$ gets closer to being supermodular, and we obtain tighter bounds on the profitability of collusion and shill bidding. Furthermore, if $\gamma_{\text{sup}}=1$, we obtain exact coalition-proofness. As is proven in Proposition~\ref{rem}-(ii), this result coincides with the one on supermodularity in \cite[Theorem~3]{karaca2017game}. 
In practice, achieving the profit in Theorem~\ref{thm:approx_collusions} may still be difficult since in general, the bidders need full information on the market constraints and other bidders to collude optimally.

Following from these discussions, a natural question is whether the function $J$ in \eqref{eq:main_model} satisfies weak supermodularity. 

\begin{theorem}\label{thm:lowbnd}
	The market objective $J$ in~\eqref{eq:main_model} is weakly supermodular. Its supermodularity ratio satisfies $\gamma_{\text{sup}}\geq 1/k_{\text{feas}} >0$, where $k_{\text{feas}}$ is the maximum number of bidders that can be removed while ensuring the feasibility of~\eqref{eq:main_model}.
\end{theorem}
\begin{proof}
	For supermodularity ratio, we aim to derive the tightest lower bound on $$\dfrac{J(S_{-K}) - J(S)}{\sum_{l \in K }  J(S_{-l}) - J(S)},$$ for all $K,S\subseteq L$. There are four possible cases to consider.\\
	
		(1) If $J(S_{-K}) - J(S)=0$ and $\sum_{l \in K }  J(S_{-l}) - J(S)=0$, then, for this instance, the supermodularity ratio is $1$.\\
		
		(2) The case where we have ${J(S_{-K}) - J(S)>0}$ and ${\sum_{l \in K }  J(S_{-l}) - J(S)=0}$ can be ignored, since~we are looking for a lower bound.\\
		
		(3) Consider the case of ${\sum_{l \in K }  J(S_{-l}) - J(S)>0}$ but ${J(S_{-K}) - J(S)=0}$. We show that this can never be the case for the model in \eqref{eq:main_model}. Since the market objective is nonincreasing we have $J(R) \leq J(S)$ for $S\subseteq R$. Then, $$J(S_{-K})\geq J(S_{-l})\geq  J(S),$$ for all $l \in K$ and $K,S\subseteq L$. As a result, we obtain $J(S_{-K})\geq J(S_{-l})\geq  J(S_{-K})$, or equivalently $J(S_{-K})=J(S_{-l})$. This observation yields $$\sum_{l \in K }  J(S_{-l}) - J(S)=\sum_{l \in K }  J(S_{-K}) - J(S)=0,$$ which contradicts the initial assumption.\\
		
		(4) Finally, we consider the case of ${J(S_{-K}) - J(S)>0}$ and ${\sum_{l \in K }  J(S_{-l}) - J(S)>0}$. This yields a positive lower bound on the supermodularity ratio. Note that if $J(S_{-l})$ is infeasible, so is $J(S_{-K})$. Hence, we can ignore such infeasible sets for computing a lower bound. Restricting our attention to the case of ${J(S_{-K}) - J(S)>0}$ and ${\sum_{l \in K }  J(S_{-l}) - J(S)>0}$, a lower bound on $\gamma_{\text{sup}}$ is given by 
		\begin{align*}
			\gamma_{\text{sup}}
			&= \min_{\substack{S,K\subseteq L\\ |K|\leq k_{\text{feas}}}} \dfrac{J(S_{-K}) - J(S)}{\sum_{l \in K }  J(S_{-l}) - J(S)}\\
			&\geq \min_{\substack{S,K\subseteq L\\ w\in K,\, |K|\leq k_{\text{feas}}}} \dfrac{J(S_{-K}) - J(S)}{|K|[J(S_{-w}) - J(S)]}\\
			&\geq \min_{\substack{S,K\subseteq L\\ w\in K,\, |K|\leq k_{\text{feas}}}} \dfrac{J(S_{-w}) - J(S)}{|K|[J(S_{-w}) - J(S)]}\\&= \dfrac{1}{k_{\text{feas}}}>0. 
		\end{align*}
		The equality follows from ignoring infeasible sets. The first inequality follows from $w$ yielding the maximum value for $J(S_{-l}) - J(S)$, over $l\in K$. The second inequality comes from $J(S_{-K})\geq J(S_{-w})$. Since ${\sum_{l \in K }  J(S_{-l}) - J(S)>0}$ and $w$ yields the maximum, we have $J(S_{-w}) - J(S)> 0$. We obtain the final equality. This concludes the proof.  \QEDA
\end{proof}

Theorem~\ref{thm:lowbnd} shows that any electricity market auction modeled by \eqref{eq:main_model} is weakly supermodular, and its supermodularity ratio is lower bounded by $1/k_{\text{feas}} \geq 1/|L|>0$. Furthermore, we obtain $\gamma_{\text{sup}}=1$ if the problem \eqref{eq:main_model} is infeasible whenever any two bidders are removed. This can be regarded as an alternative proof to \cite[Proposition~1]{karaca2017game}. However, the lower bound we derived is often conservative. In Section~\ref{sec:num}, we obtain larger supermodularity ratios from studies based on realistic electricity market instances. In the numerics, we also discuss a computationally efficient method for computing the supermodularity ratio.

It is worth mentioning the previous research on lower bounding the submodularity ratio. The work in \cite{elenberg2016restricted} lower bounds the submodularity ratio of an unconstrained optimization problem using restricted strong convexity and restricted smoothness of the objective function. These bounds on submodularity ratio can potentially carry over to the supermodularity ratio of an unconstrained problem. By~contrast, we consider a constrained problem, and hence their results do not extend to the supermodularity ratio of problem \eqref{eq:main_model}.

In summary, the VCG mechanism satisfies the approximate coalition-proofness property for the electricity markets modeled by~\eqref{eq:main_model}. Subsequently, we can provide bounds on the profitability of collusion and shill bidding. In addition, better approximate coalition-proofness properties are achieved if the market objective $J$ is close to being supermodular.

\section{Numerical Results} \label{sec:num}
Our goal is to show the effectiveness of the supermodularity ratio to predict collusion potential in electricity markets. We start by providing a simple example to show that the upper bound in Theorem~\ref{thm:approx_collusions} and the lower bound in Theorem~\ref{thm:lowbnd} are attained. Then, we consider IEEE test systems with DC power flow models. These markets do not satisfy the stringent conditions for supermodularity when line limits are present. We calculate the supermodularity ratios of these systems and show that they are close to~$1$. Consequently, we obtain tight bounds on collusion and shill bidding.

Before discussing the numerical results, we explain the computation method for the supermodularity ratio in Definition~\ref{def:suprat}.
To calculate $\gamma_{\text{sup}}$, we need to solve up to $2^{|L|}$ instances of the market problem~\eqref{eq:main_model} and then we could form the following linear program with up to $2^{2|L|}$ linear constraints,
\begin{equation}\label{eq:supratprob}
	\begin{split}
	\gamma_{\text{sup}}=&\argmax_{\gamma}\, \gamma\\
	&\ \ \,\mathrm{s.t.}\ \ \gamma\Big[\sum_{l \in K }  J(\mathcal{B}_{S\setminus{l}}) - J(\mathcal{B}_{S})\Big]\leq J(\mathcal{B}_{S\setminus{K}}) - J(\mathcal{B}_{S}),\\
	\end{split}
\end{equation}
 for all $K,S\subseteq L$.
 To deal with the exponential size, we use the constraint generation approach proposed in \cite{day2007fair}. The method proceeds as follows.  
 We initialize the algorithm by setting $\gamma_{\text{sup}}=1$. For this candidate solution, we can formulate another problem that finds the constraint in \eqref{eq:supratprob} with the largest violation, that is, the largest $\gamma_{\text{sup}}[\sum_{l \in K }  J(\mathcal{B}_{S\setminus{l}}) - J(\mathcal{B}_{S})]- [J(\mathcal{B}_{S\setminus{K}}) - J(\mathcal{B}_{S})].$ For this formulation, we refer to~\cite{karaca2017game}.
  We then obtain another candidate solution by finding the largest ratio that satisfies this violated constraint. The algorithm iterates between two problems and converges to the supermodularity ratio~\cite[Theorem~4.2]{day2007fair}.  This algorithm may still require the generation of all constraints, but in practice, it converges fast. Note that, this problem also needs to be solved under all bid profiles. We tackle this~by solving~\eqref{eq:supratprob} for many randomly generated bid profiles and setting the supermodularity ratio as the minimum of these~ratios. 

\subsection{A simple example for Theorems~\ref{thm:approx_collusions} and \ref{thm:lowbnd}}
	Let $\epsilon$ be a small positive number. Suppose the central operator has to procure $800$ MW of power supply from bidders $1$, $2$ and $3$ who have the true costs $\$600$ for $800$ MW, $\$300+\epsilon$ for $400$ MW and $\$300+\epsilon$ for $400$ MW, respectively. Under the VCG mechanism, the dominant-strategy Nash equilibrium is truthful bidding. Hence, assume all the bidders are truthful. As a result, bidder~$1$ wins and receives $p_1^{\text{VCG}} = 600 + (600+2\epsilon-600)= \$600+2\epsilon$. Invoking Theorem~\ref{thm:lowbnd}, we obtain the lower bound $\gamma_{\text{sup}}\geq 1/2$ since $k_{\text{feas}}=2$. In fact, this lower bound is tight. We can verify that the supermodularity ratio of this market is given by $\gamma_{\text{sup}}=1/2$ by evaluating the constraints in Definition~\ref{def:suprat} under any set of bid prices.
	
	Now, suppose bidders~$2$ and~$3$ collude and change their bids to $\$0$ for $400$ MW. Then, bidders $2$ and $3$ win and receive a payment of $\$600$ each. The total VCG utility they receive is $\$600-2\epsilon$. In this case, the upper bound in Theorem~\ref{thm:approx_collusions} is tight since  $[\gamma_{\text{sup}}^{-1}-1][J(\mathcal{C})-J(\mathcal{C}_{-K},\mathcal{B}_{K}^0)]=[1/(1/2)-1][600-0]=\$600$. For this example, we highlight that, for bidders~$2$ and~$3$, lowering bid prices to zero results in the largest profit these bidders can achieve by collusion. However, in general, it is not straightforward how to optimally collude since the bidders do not have full information on the market constraints and other bidders.

\subsection{IEEE test systems with DC power flow models}

In wholesale electricity markets, the central operator's problem involves a power grid model with network balance constraints. In this section, we adopt the DC power flow model~\cite{wu1996folk}. This model assumes lossless lines, constant bus voltages, and small phase angle differences. Under the LMP mechanism, the generators may act strategically to manipulate the payments in these markets~\cite{wolfram1997strategic,joskow2001quantitative}. The works in~\cite{xu2017efficient,karaca2017game} proposed the VCG mechanism as a way to handle this issue. However, these markets have quadratic bids and polytopic constraints. 
Even though the bid curves are convex, the polytopic constraint set of a DC power flow is not a polymatroid whenever the line limits are present. For this reason, the VCG mechanism can suffer from collusion and shill bidding \cite{karaca2017game}. 

Here, we consider the IEEE $14$-bus, $30$-bus, and $118$-bus test systems in \cite{christie2000power}.  In~\cite[\S 5.2]{karaca2017game}, we analyzed the VCG payments and how they compare with the payments for the LMP and the pay-as-bid mechanisms. In this work, we focus on the supermodularity ratio to quantify coalition-proofness of each system. To identify the collusion potential under any market instance, the supermodularity ratio computations are further verified by choosing the coefficients of the quadratic bids from a uniform distribution.

For the $14$-bus example, we have $10$ MW limits on lines exiting node~$1$ and connecting it to nodes~$2$ and~$5$. Even though the bids are convex, the polytopic constraint set is not a polymatroid and conditions for supermodularity do not hold. Invoking Theorem~\ref{thm:lowbnd} we obtain a lower bound $\gamma_{\text{sup}}\geq1/2$ since $k_{\text{feas}}=2$.  In contrast, the estimated supermodularity ratio, calculated using the constraint generation approach discussed above, is $0.76$. Next, we consider the bid curves in~\cite{christie2000power}. Since all bidders are allocated a positive quantity, we add two losing bidders at nodes~$1$ and~$5$. The supermodularity ratio indicates that losing bidder's collective profit from collusion would be upper bounded by $\$1296$, which is $11\%$ of the total VCG payment. Similar bounds are obtained for other possible colluders. We conclude that Theorem~\ref{thm:approx_collusions} provides us with bounds on strategic manipulations.

Similarly, the $30$-bus system also has polytopic constraints and the conditions for supermodularity do not hold. By contrast, we calculated the supermodularity ratio to be $1$. This verifies the exact coalition-proofness of the $30$-bus system under the VCG mechanism. This result can be explained in two ways. First, none of the line limits are tight at the optimal solution. Second, removing two bidders yields an infeasible problem in most cases.

In the case of the $118$-bus system,  the constraint set is a polymatroid because there are no line limits in the model~\cite{christie2000power}. It follows that the market objective is supermodular~\cite[Theorem~5]{karaca2017game}. Additionally, supermodularity ratio is calculated to be $1$, as is proven in Proposition~\ref{rem}-(ii). For the bid curves in~\cite{christie2000power}, there are two losing generators located at nodes $1$ and~$4$. Suppose these two losing generators form a coalition and lower their bids to zero. Then, their collective profit decreases from $\$0$ to $-\$439.8$. Invoking Theorem~\ref{thm:approx_collusions}, collusion is not profitable for any set of colluders who lose when bidding their true costs. As a final remark, we fix $50$~MW limits on two lines, one connecting nodes~$5$ and~$6$, another connecting nodes~$9$ and~$10$. Then, the supermodularity ratio is calculated to be~$0.92$. This shows that we can obtain weak supermodularity by introducing line limits, similar to the $14$-bus system.

\section{Conclusion}
For the electricity markets, the incentive-compatible VCG mechanism was susceptible to collusion and shill bidding since the market objective was in general not supermodular. Motivated by this, we defined the supermodularity ratio to quantify how close a set function is to being supermodular. The supermodularity ratio of the market objective provided us with bounds on the profitability of collusion and shill bidding. These bounds get tighter as the ratio increases. We then derived an analytical lower bound on the supermodularity ratio of the electricity markets. The results derived apply to general auctions run by the VCG mechanism. By quantifying coalition-proofness, we can evaluate the applicability of the VCG mechanism in terms of collusion and shill bidding.
Finally, we illustrated the tightness of the bounds on supermodularity ratio, and verified our results with case studies based on the IEEE test systems.

Our future work will address deriving bounds on collusion and shill bidding for double-sided auctions using the idea introduced here. 

\appendix
\section{Proof of Proposition~\ref{rem}}
\begin{proof}
	For a nonincreasing set function, we prove the statements in Proposition~\ref{rem}. 
	\begin{itemize}
		\item[(i)] Since the function is nonincreasing, we observe that both sides of the equation are always nonnegative in Definition~\ref{def:suprat}. Since $\gamma_{\text{sup}}$ is the largest scalar, we have $\gamma_{\text{sup}}\geq 0$. Furthermore, the inequality must hold when $K$ is a singleton. Then, we obtain $\gamma_{\text{sup}}\leq 1$.
		\item[(ii)] The set function is supermodular if and only if the supermodularity ratio is given by $\gamma_{\text{sup}}=1$. 
		
		($\Longrightarrow$) We first prove that supermodularity implies $\gamma_{\text{sup}}=1$. 
		Let $K=\{ l_1, \dots, l_k\}$. 
		Notice that, by supermodularity,
		$$
		J({S_{-{l_i}}}) 
		- J(S) 
		\leq 
		J({S_{-\{ l_i ,..., l_k \}} })
		-
		J({S_{-\{ l_{i+1} ,..., l_k \}}}). 
		$$
		Thus, we have 
		\begin{equation*}
		\begin{split}
		\sum_{l \in K }  J(S_{-l}) - J(S) &=\sum_{i=1}^{k}J(S_{-{l_i}}) - J(S)\\
		&\leq\sum_{i=1}^{k}J(S_{-{\{l_i,\ldots,l_k\}}}) - J(S_{-{\{l_{i+1},\ldots,l_k\}}})\\
		&= J(S_{-K}) - J(S).\\
		\end{split}
		\end{equation*}
		The first equality follows from the definition of the set~$K$. Supermodularity implies the inequality. The last equality holds by a telescoping sum. Finally, the same arguments can be made for all $K,S\subseteq L$.
		Since the supermodularity ratio is the largest scalar such that this inequality holds, we obtain $\gamma_{\text{sup}}\geq1$. Combining it with $\gamma_{\text{sup}}\leq 1$ from part~(i), we obtain $\gamma_{\text{sup}}=1$.
		
		($\Longleftarrow$) To prove that the supermodularity is also necessary for $\gamma_{\text{sup}}=1$, we proceed by contradiction.
		Suppose supermodularity does not hold. Then, there exist sets ${S}\subseteq{R}$, and $l\in S$ such that $J({R_{-l}})\! -\! J(R)\!>\!  J({S_{-l}})\! -\! J(S)$. First, we show that without loss of generality, we can restrict $R$ to differ from $S$ by one bidder, that is, $R= S\cup \{i\}$ for some $i$. We take $S^0=S$ and
		$S^j=S^{j-1}\cup \{ l_j\}$ with $S^k=R$, then,
		\begin{align*}
		\sum_{j =1}^k J({S^j_{-l}})-J({S^{j-1}_{-l}}) &= J({R_{-l}})-J({S_{-l}})\\
		&>  J(R)-J(S)\\ &= \sum_{j=1}^k J({S^j})-J({S^{j-1}}).
		\end{align*}
		The strict inequality above must hold for one of the summands, that is, $\exists j,\,J({S^j_{-l}})-J({S^{j-1}_{-l}}) > J({S^j})-J({S^{j-1}})$.  Equivalently, $\exists j,\,J({S^j_{-l}})-J({S^j})>J({S^{j-1}_{-l}})-J({S^{j-1}})$. Hence, we can consider sets ${S} \subseteq {R}$ that differ only by bidder~$i$. By this observation, we have
		\begin{equation}\label{eq:side}
		\begin{split} J({{R}_{-l}}) - J({R}) &> J({{S}_{-l}}) - J({S})\\ &= J({{R}_{-{\{i,l\}}}}) - J({{R}_{-i}})\\& \geq 0,\end{split}
		\end{equation} for $i \in {R}\setminus {S}$.
		Considering the set ${R}$, and ${K}=\{i, {l}\}$, we have: 
		{
			\begin{align*}
			\sum_{l' \in {K}}\!\! J({{R}_{-l'}}) - J({R})
			=&
			J({{R}_{-{l}}}) - J({{R}}) + J({{R}_{-i}}) - J({R}) \\ 
			>& 
			J({{R}_{-{\{i,l\}}}}) - J({{R}_{-i}}) + J({{R}_{-i}}) - J({R}) 
			\\
			=& J({{R}_{-K}}) - J({R}).
			\end{align*}
			The strict inequality follows from \eqref{eq:side}.
			If the supermodularity ratio was $\gamma_{\text{sup}}=1$, then the inequality in Definition~\ref{def:suprat}  would be violated for this particular choice of $R,K\subseteq L$. Hence, the supermodularity ratio is $\gamma_{\text{sup}}<1$. Thus, we conclude that the supermodularity condition is also necessary for $\gamma_{\text{sup}}=1$.
		}
		
	\end{itemize}
	This concludes the proof of Proposition~\ref{rem}. \QEDA
\end{proof}

\bibliographystyle{IEEEtran}
\bibliography{IEEEabrv,library}
\end{document}